\documentclass[conference]{IEEEtran}
\IEEEoverridecommandlockouts
\usepackage{cite}
\usepackage{algorithmic}
\usepackage{graphicx}
\usepackage{textcomp}

\usepackage[utf8]{inputenc} 
\usepackage[T1]{fontenc}    
\usepackage{url}            
\usepackage{booktabs}       
\usepackage{nicefrac}       
\usepackage{microtype}      

\usepackage[cmex10]{amsmath}
\usepackage{amsfonts,amssymb,amsthm,bbm,verbatim}

\usepackage[usenames,dvipsnames]{color}
\usepackage[T1]{fontenc}
\usepackage{mathtools}
\usepackage{siunitx,caption}
\usepackage{subfig}
\usepackage{threeparttable}
\usepackage{dsfont}
\usepackage{placeins}
\usepackage{tikz,pgfplots,pgfplotstable,xcolor}
\usetikzlibrary{matrix}
\usetikzlibrary{cd}
\usetikzlibrary{calc}
\usetikzlibrary{arrows}      
\usepgflibrary{arrows}
\usetikzlibrary{arrows.meta}
\usetikzlibrary{decorations.markings}
\usepackage{enumitem}

\DeclareMathOperator{\E}{\mathbb{E}}
\DeclareMathOperator*{\argmin}{arg\,min}

\DeclareMathOperator*{\conv}{conv}

\newcommand{\Xcal}{\mathcal{X}}
\newcommand{\Ycal}{\mathcal{Y}}
\newcommand{\Zcal}{\mathcal{Z}}

\def\P{{\mathbb{P}}}

\newcommand{\TCI}{\widetilde{CI}}
\newcommand{\TSI}{\widetilde{SI}}
\newcommand{\TUI}{\widetilde{UI}}

\newcommand{\bracketed}[1]{\left[\vphantom{\let\\=\relax #1}\right. #1 \left.\vphantom{\let\\=\relax #1}\right]}

\usepackage{blkarray}

\newcommand{\ptrain}{\hat{p}_{\textrm{data}}}

\newcommand{\KL}{D_{\mathrm{KL}}}

\theoremstyle{definition}
\newtheorem{theorem}{Theorem}
\newtheorem{lemma}[theorem]{Lemma}
\newtheorem{proposition}[theorem]{Proposition}

\newtheorem{definition}[theorem]{Definition}
\newtheorem{example}[theorem]{Example}

\def\BibTeX{{\rm B\kern-.05em{\sc i\kern-.025em b}\kern-.08em
		T\kern-.1667em\lower.7ex\hbox{E}\kern-.125emX}}

\graphicspath{{./figures/}}

\begin{document}
	
\setlength{\abovedisplayskip}{4pt}
\setlength{\belowdisplayskip}{4pt}

\title{The Variational Deficiency Bottleneck\\
\thanks{This project has received funding from the European Research Council (ERC) under the European Union's Horizon 2020 research and innovation programme (grant agreement n\textsuperscript{o} 757983).}
}

\author{%
	\IEEEauthorblockN{Pradeep Kr.~Banerjee}
	\IEEEauthorblockA{MPI MiS\\ 
		Email: pradeep@mis.mpg.de}
	\and
	\IEEEauthorblockN{Guido Mont{\'u}far}
	\IEEEauthorblockA{UCLA and MPI MiS\\ 
		Email: montufar@math.ucla.edu}
}

\maketitle

\begin{abstract}
We introduce a bottleneck method for learning data representations based on information deficiency, rather than the more traditional information sufficiency. A variational upper bound allows us to implement this method efficiently. The bound itself is bounded above by the variational information bottleneck objective, and the two methods coincide in the regime of single-shot Monte Carlo approximations. The notion of deficiency provides a principled way of approximating complicated channels by relatively simpler ones. We show that the deficiency of one channel with respect to another has an operational interpretation in terms of the optimal risk gap of decision problems, capturing classification as a special case. Experiments demonstrate that the deficiency bottleneck can provide advantages in terms of minimal sufficiency as measured by information bottleneck curves, while retaining 
robust test performance in classification tasks.
\end{abstract}


\vspace{.3cm}
\section{Introduction}
The information bottleneck (IB) is an approach to learning data representations based on a notion of minimal sufficiency. The general idea is to map an input source to an intermediate representation that retains as little information as possible about the input (\emph{minimality}), but preserves as much information as possible in relation to a target variable of interest (\emph{sufficiency}). See Fig.~\ref{figparadig}. For example, in a classification problem, the target variable could be the class label of the input data. In a reconstruction problem, the target variable could be a denoised reconstruction of the input. Intuitively, a representation which is minimal in relation to a given task, will discard nuisances in the inputs that are irrelevant to the task, and hence distill more meaningful information and allow for a better generalization. 
The IB methods~\cite{witsenhausen1975,InformationBottleneck,Harremoes2007,GIB2018} have found numerous applications in representation learning, clustering, classification, generative modeling, model selection and analysis in deep neural networks, among others (see, e.g., \cite{Tishbylearning,betaVAE,alemiVAE,tishbydeep,tishbyopen,zivPolyanskiIB}). 

In the traditional IB paradigm, minimality and sufficiency are measured in terms of the mutual information. Computing the mutual information can be challenging in practice. 
Recent works have formulated more tractable functions by way of variational bounds on the mutual information~\cite{chalk2016,achilledropout,alemi2016variational,artemyNIB}. 
Instead of maximizing the sufficiency term of the IB, we formulate a new bottleneck method that minimizes deficiency. Deficiencies provide a principled way of approximating complex channels by relatively simpler ones and have a rich heritage in the theory of comparison of statistical experiments~\cite{Blackwell1953,lecam,torgersen}. From this angle, the formalism of deficiencies has been used to obtain bounds on optimal risk gaps of statistical decision problems. As we show, the deficiency bottleneck minimizes a regularized risk gap. Moreover, the proposed method has an immediate variational formulation that can be easily implemented as a modification of the variational information bottleneck (VIB)~\cite{alemi2016variational}. In fact, both methods coincide in the limit of single-shot Monte Carlo approximations. We call our method the variational deficiency bottleneck (VDB). 

Experiments on basic data sets show that the VDB is able to obtain more compressed representations than the VIB while retaining the same level of sufficiency.
Training with the VDB also improves out-of-distribution robustness over the VIB as we demonstrate on two benchmark datasets, the \textsc{MNIST-C} \cite{mnistc} and the \textsc{CIFAR-10-C} \cite{augmixICLR2019}.

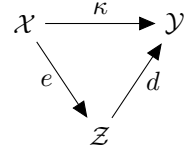
\begin{figure}[t!]
	\label{fig:VDB} 
	\centering
	\begin{tikzpicture}[baseline=(current bounding box.center),main node/.style={circle,draw}]
	\node	(1) at	(0,0)  {$\Xcal$}; 
	\node[]	(2)	 at	(2,0) {$\Ycal$};  
	\node	(3)	 at	(1,-1.5) {$\Zcal$};
	\draw[->,above,arrows={-triangle 45}] (1) to node {$\kappa$} (2);
	\draw[->, left,arrows={-triangle 45}] (1) to node  {$e$} (3);
	\draw[->,below, right,arrows={-triangle 45}] (3) to node [xshift=0pt,yshift=0pt] {$d$} (2);
	\end{tikzpicture}
	\caption{The bottleneck paradigm: The general idea of a bottleneck method is to first map an input~$X\in\Xcal$ to an intermediate representation~$Z\in\Zcal$, and then map~$Z$ to an output~$Y\in\Ycal$. We call the mappings, resp., an encoder ($e$) and a decoder ($d$). In general, the true channel~$\kappa$ is unknown, and only accessible through a set of training examples. We would like to obtain an approximation of~$\kappa$. 	 
	}
	\label{figparadig}
\end{figure}

We describe the details of our method in Section~\ref{sec:method}. We elaborate on the theory of deficiencies in Section~\ref{sec:VDB_theory}. Experimental results with the VDB are presented in Section~\ref{sec:experiments}. We use notation that is standard in information theory~\cite{CoverThomas91:Elements_of_Information_Theory}. 

\vspace{.3cm}
\section{The variational deficiency bottleneck} 
\label{sec:method}
Let~$X$ denote an observation or \emph{input} variable and~$Y$ an \emph{output} variable of interest and let $\Xcal$, $\Ycal$ denote, resp., the space of possible inputs and outputs. Let $p(x,y)=\pi(x)\kappa(y|x)$ be the true joint distribution, where the conditional distribution or \emph{channel}~$\kappa(y|x)$ describes how the output depends on the input. We consider the situation where the true channel is unknown, but we are given a set of~$N$ independent and identically distributed (i.i.d.) samples $(x^{(i)},y^{(i)})_{i=1}^N$ from~$p$. Our goal is to use this data to learn a more structured version of the channel~$\kappa$, by first ``compressing'' the input~$X$ to an intermediate \emph{representation} variable~$Z$ and subsequently mapping the representation back to the output $Y$. The presence of an intermediate representation can be regarded as a bottleneck, a model selection problem, or as a regularization strategy. 

We define an encoder and a decoder model using two parameterized families of channels $e(z|x)$ and $d(y|z)$, respectively. The encoder-decoder pair induces a model~$\widehat{\kappa}(y|x)=\int d(y|z)e(z|x)\,dz$. Equivalently, we write~$\widehat{\kappa}=d\circ e$. Given a representation, we want the decoder to be as powerful as the original channel~$\kappa$ in terms of ability to recover the output. The deficiency of a decoder~$d$ w.r.t.~$\kappa$ quantifies the extent to which any pre-processing at the input (by way of randomized encodings) fails to approximate~$\kappa$. 
Let $\mathsf{M}(\Xcal;\Ycal)$ denote the space of all channels from $\Xcal$ to~$\Ycal$. We define the deficiency of $d$ w.r.t. $\kappa$ as follows:
\begin{definition}[Deficiency] \label{def:deficiency}
	Given a channel~$\kappa\in\mathsf{M}(\Xcal;\Ycal)$ from~$\Xcal$ to~$\Ycal$, and a decoder $d\in\mathsf{M}(\Zcal;\Ycal)$ from some~$\Zcal$ to~$\Ycal$, the \emph{deficiency of $d$ w.r.t.~$\kappa$} is defined as
	\begin{align}
	\delta^\pi(d,\kappa) = \min_{e  \in \mathsf{M}(\Xcal; \Zcal)} \KL(\kappa \| d\circ e|\pi). 
	\end{align}
\end{definition}
Here~$\KL(\cdot\|\cdot|\cdot)$ is the conditional KL divergence~\cite{CoverThomas91:Elements_of_Information_Theory}, and $\pi$ is an input distribution over $\Xcal$. 
The definition is similar in spirit to Lucien Le Cam's notion of weighted deficiencies of one channel w.r.t.~another~\cite{lecam}, \cite[Section 6.2]{torgersen} and its recent generalizations \cite{raginsky2011}. 

We propose to train the model by minimizing the deficiency of $d$~w.r.t.~$\kappa$ subject to a regularization that limits the \emph{rate}~$I(Z;X)$, i.e., the mutual information between the representation and the raw inputs. We call our method the \emph{deficiency bottleneck} (DB). The DB minimizes the following objective over all tuples~$(e\in\mathsf{M}(\Xcal;\Zcal),d\in\mathsf{M}(\Zcal;\Ycal))$: 
\begin{align}\label{eq:DB-obj}
\mathcal{L}_{DB}^{\beta}(e,d) := \delta^\pi(d,\kappa) + \beta I(Z;X).  
\end{align}
The parameter $\beta\geq0$ allows us to adjust the level of regularization. 

For any distribution~$r(z)$, the rate term admits a simple variational upper bound \cite{CoverThomas91:Elements_of_Information_Theory}: 
\begin{align}\label{eq:varboundMI}
I(Z;X) \le \int p(x,z)\log{\tfrac{e(z|x)}{r(z)}}\mathop{dx}\mathop{dz}.
\end{align}

Let~$\ptrain$ be the empirical distribution of the data (input-output pairs). By noting that~$\delta^\pi(d,\kappa)\leq \KL(\kappa\| d\circ  e|\pi)$ for any $e\in\mathsf{M}(\Xcal;\Zcal)$, and ignoring data-dependent constants, we obtain the following optimization objective which we call the \emph{variational deficiency bottleneck} (VDB) objective:
\begin{multline}
\mathcal{L}_{VDB}^{\beta}(e,d) := \mathbb{E}_{(x,y)\sim \ptrain} \bracketed{-\log \int d(y|z)e(z|x)dz \\ + \beta \KL(e(Z|x)\| r(Z)) }. \label{eq:VDB-obj}
\end{multline} 

The computation is simplified by defining~$r(z)$ to be a standard multivariate Gaussian distribution~$\mathcal{N}(0,I)$, and using an encoder of the form~$e(z|x) = \mathcal{N}(z|f_\phi(x))$, where $f_\phi$ is a neural network that outputs the parameters of a Gaussian distribution. Using the reparameterization trick~\cite{VAEkingma,VAErezende}, we then write~$e(z|x) dz = p(\epsilon)d\epsilon$, where $z=f(x,\varepsilon)$ is a function of $x$ and the realization $\epsilon$ of a standard normal distribution. This allows us to do stochastic backpropagation through a single sample~$z$. The KL term in~\eqref{eq:VDB-obj} admits an analytic expression when $r(z)$ and the encoder are Gaussian. We train the model 
by minimizing the following empirical objective over all tuples~$(e\in\mathsf{M}(\Xcal;\Zcal),d\in\mathsf{M}(\Zcal;\Ycal))$: 
\begin{multline}
\frac{1}{N} \sum_{i=1}^N  \bracketed{-\log( \frac{1}{M} \sum_{j=1}^M [ d(y^{(i)}|f(x^{(i)},\epsilon^{(j)}))])\\ + \beta \KL(e(Z|x^{(i)})\|r(Z))}. \label{eq:emp-VDB-obj}
\end{multline}
For training, we choose a mini-batch size of~$N=100$. For estimating the expectation inside the log, we use~$M=1,3,6,12$ Monte Carlo samples from the encoding distribution. 

We note that the Variational Information Bottleneck (VIB)~\cite{alemi2016variational} leads to a similar-looking objective function, with the only difference that the sum over~$j$ in \eqref{eq:emp-VDB-obj} is \emph{outside} of the logarithm. By Jensen's inequality, the VIB loss is an upper bound to our loss. If one uses a single sample from the encoding distribution (i.e., $M=1$), the VDB and the VIB empirical objective functions coincide. For a large enough mini-batch size, e.g., $N=100$, taking $M = 1$ is sufficient to estimate the VIB objective \cite{alemi2016variational}. This is the standard setting as presented in \cite{alemi2016variational} that we want to compare with. In the case of the VDB, on the other hand, the mini-batch size $N$ and $M$ are \emph{not} exchangeable, since the expectation is inside the log function.

To better understand the behavior of the VDB optimization \eqref{eq:emp-VDB-obj}, we adopt two training strategies: 
\begin{itemize}
	\item a \emph{oneshot} strategy where the encoder and decoder networks are updated simultaneously, and
	\item a \emph{sequential} strategy where the encoder network is updated for $k$ steps before alternating to one decoder update. We choose $k=5,10,20$. 
\end{itemize}
The idea of using the sequential strategy is to better approximate the deficiency which involves an optimization over the encoder (see Definition~\ref{def:deficiency}).

\vspace{.3cm}
\section{Two perspectives on the deficiency bottleneck} \label{sec:VDB_theory}
In this section, we present two different perspectives on the deficiency bottleneck, namely, a decision-theoretic perspective and an information decomposition perspective.  

In Section~\ref{sec:defriskgap}, we 
review the notions of information sufficiency and deficiency through the lens of Blackwell-Le Cam decision theory \cite{Blackwell1953,lecam,torgersen}. We formulate the learning task as a decision problem and give an operational characterization of the deficiency $\delta^{\pi}(d,{\kappa})$ as the gap in the expected losses of optimal decision rules when using the channel $d$ rather than~$\kappa$. 

In Section~\ref{sec:UIB}, we review the classical IB and our DB objective through the lens of nonnegative mutual information decompositions \cite{e16042161,UIdefAllerton,HarderSalgePolani2013:Bivariate_redundancy}. This leads us to a new interpretation of the IB as a \emph{Unique Information Bottleneck} and also sheds light on the difference between the IB and DB formulations.

\subsection{A decision-theoretic perspective} \label{sec:defriskgap}
\subsubsection{Blackwell sufficiency and channel deficiency} 
In a seminal paper \cite{Blackwell1953}, David Blackwell asked the following question: Suppose that a learner has a finite set of possible actions and she wishes to make an optimal decision 
to minimize a loss depending on the value of some random variable~$Y$ and her chosen action. 
If the learner cannot observe $Y$ directly before choosing her action and has to pick between two channels with the \emph{common input} $Y$, which one should she prefer? 
Blackwell introduced an ordering that compares channels by the minimal expected loss or risk that a learner incurs when her decisions are based on the channel outputs. He then showed that such an ordering can be equivalently characterized in terms of a purely probabilistic relation between the channels: The learner will always prefer one channel over another if and only if the latter is an \emph{output-degraded} version of the former, in the sense that she can simulate a single use of the latter by randomizing at the output of the former. 

Very recently, Nasser \cite{nasser2017} asked the same question, only now the learner has to choose between two channels with a \emph{common output} alphabet. 
Nasser introduced the \emph{input-degraded} ordering and gave a characterization of input-degradedness that is similar to Blackwell's ordering \cite{Blackwell1953}. 
\begin{definition}[Blackwell sufficiency]\label{def:preorder}
	Given two channels, $\kappa\in\mathsf{M}(\Xcal; \Ycal)$ and~$d\in\mathsf{M}(\Zcal; \Ycal)$, ${\kappa}$ is \emph{input-degraded} from~${d}$, denoted ${d}\succeq_{\Ycal} {\kappa}$, if~${\kappa}={d}\circ {e}$ for some~${e}\in\mathsf{M}(\Xcal;\Zcal)$. 
	We say that ${d}$ is \emph{input Blackwell sufficient} for~$\kappa$ if~${d}\succeq_{\Ycal} {\kappa}$. 
\end{definition}
Stated in another way, $d$ is input Blackwell sufficient for $\kappa$ if ${d}$ can be reduced to~$\kappa$ by applying a randomization $e$ at its input so that $d\circ e=\kappa$. Blackwell sufficiency induces only a preorder on the set of all channels with a common output alphabet. In practice, most channels are uncomparable, i.e., one cannot be reduced to another by a randomization. When such is the case, the deficiency quantifies how far the true channel~$\kappa$ is from being a randomization (by way of any input encoding $e$) of the decoder~$d$.

\subsubsection{Deficiency as an optimal risk gap} 
We formulate a learning task as a decision problem and show that the deficiency quantifies the gap in the optimal risks 
when using the channel $d$ rather than~$\kappa$. 

Let $\P_Y$ be the set of all distributions on $\Ycal$. 
In the following, we assume that $\Xcal$ and $\Ycal$ are finite. 
For every $x \in \Xcal$, define $\kappa_x \in \P_Y$ as $\kappa_x(y)=\kappa(y|x), \text{ } \forall y\in\Ycal$. 
Consider the following decision problem between a learner and Nature: 
Nature draws~$x\sim \pi$ and $y\sim \kappa_x$. The learner observes~$x$ and proposes a distribution~$q_x\in\mathbb{P}_{\Ycal}$ that expresses her uncertainty about the true value~$y$. The quality of a prediction~$q_x$ in relation to~$y$ is measured by the \emph{log-loss} function~$\varsigma(y,q_x) :=-\log{q_x(y)}$. 
The log-loss is an instance of a ``strictly proper'' loss function that enjoys nice properties such as the uniqueness of the optimum; 
see,  e.g.,~\cite{gneiting}. 

Ideally, the prediction~$q_x$ should be as close as possible to the true conditional distribution~$\kappa_x$. This is achieved by minimizing the expected loss~$L(\kappa_x,q_x):=\E_{y\sim \kappa_x}\varsigma(y,q_x)$, for all~$x\in\Xcal$. 
Define the \emph{Bayes act} against~$\kappa_x$ as the optimal prediction 
$q^{\ast}_x:= \argmin_{q_x\in\mathbb{P}_{\Ycal}} L(\kappa_x,q_x)$,
and the \emph{Bayes risk} for the distribution~$P_{XY}=\pi\times\kappa$ as $R(P_{XY},\varsigma):=\E_{x\sim\pi}L(\kappa_x,{q^{\ast}_x})$. 
For the log-loss, the Bayes act is~$q^{\ast}_x=\kappa_x$ and hence, the Bayes risk is
\begin{align*}
R(P_{XY},\varsigma)
=\E_{x\sim\pi} \E_{y\sim \kappa_x}\big[-\log{\kappa_x(y)}\big]=H(Y|X).
\end{align*}

Given a channel $d\in\mathsf{M}(\Zcal; \Ycal)$,  
we want a representation $z\in\Zcal$ of $x$ (output by some encoder), so that the outputs of $d$ match those of the true channel $\kappa$. 
Let~$C=\conv(\{d_z:\;z\in\Zcal\})\subset\mathbb{P}_{\Ycal}$ be the convex hull of the points $\{d_z\}_{z\in\Zcal}\in\mathbb{P}_{\Ycal}$. 
The Bayes act against $\kappa_x$ is
$q^{\ast}_{x_d}:=\argmin_{q_x\in C} \E_{y\sim \kappa_x}\big[-\log{q_x(y)}\big]$.
$q^{\ast}_{x_d}$ is the $rI$-projection of~$\kappa_x$ to the convex set $C\subset\mathbb{P}_{\Ycal}$ \cite{csiszarIproj}. 
Such a projection exists but is not necessarily unique. If non-unique, we arbitrarily select one of the minimizers as the Bayes act. The associated Bayes risk is 
\begin{align*}
R_d(P_{XY},\varsigma):=\E_{x\sim\pi} \E_{y\sim \kappa_x}\big[-\log{q^{\ast}_{x_d}(y)}\big]. 
\end{align*}

The next Proposition~\ref{prop:defriskgap} states that the gap in the Bayes risks, $\Delta R:=R_d(P_{XY},\varsigma)-R(P_{XY},\varsigma)$,  when making a decision based 
on $Z$ vs.\ $X$ is just the deficiency. 
\begin{proposition}[Deficiency quantifies the optimal risk gap for the log-loss]
	\label{prop:defriskgap}
	$\delta^{\pi}(d,\kappa)=\Delta R$. 
\end{proposition}
\begin{proof}
The proof follows from noting that
\begin{align*} 
\Delta R&=\sum_{x\in\Xcal} \pi(x)\min_{q_x\in C\subset\mathbb{P}_{\Ycal}} \KL({\kappa}_x\|q_x)\\&=\min_{{e}\in\mathsf{M}(\Xcal;\Zcal)} \sum_{x\in\Xcal} \pi(x) \KL({\kappa}_x\|d\circ e_x)\notag\\&=\min_{e  \in \mathsf{M}(\Xcal; \Zcal)} \KL(\kappa \| d\circ e|\pi)=\delta^{\pi}(d,{\kappa}).
\end{align*} 
\end{proof}

\subsection{An information decomposition perspective} \label{sec:UIB}
\subsubsection{IB as Unique Information Bottleneck} \label{ssec:UIB}
A quantity that is similar in spirit to the deficiency is the Unique Information ($UI$) \cite{e16042161}:
\begin{definition}[Unique information]
	Let $(Y,X,Z)\sim P$. The unique information that $X$ conveys about $Y$ w.r.t.\ $Z$ is 
	\begin{align}
	UI(Y;X\backslash Z) := \min_{Q \in \Delta_{P}} I_Q(Y;X|Z), 
	\label{subeq:UIy}
	\end{align} 
	where the subscript $Q$ 
	denotes the joint distribution on which the mutual information is evaluated, and 
	\begin{align}
	\Delta_{P} := \big\{Q \in \mathbb{P}_{\Ycal\times\Xcal\times\Zcal}\colon  &Q_{YX}(y,x)=P_{YX}(y,x),\notag 
	\\&Q_{YZ}(y,z)=P_{YZ}(y,z)\big\}
	\label{eq:delP}
	\end{align}
	is the set of joint distributions of $(Y,X,Z)$ that have the same marginals on $(Y,X)$ and $(Y,Z)$ as $P$. 
\end{definition}

While the deficiency quantifies a deviation from the \emph{input-degraded} order, the $UI$ quantifies a deviation from the \emph{output-degraded} order \cite{UIdefAllerton}.
Note, however, that the vanishing sets of $\delta^\pi(d,\kappa)$ and $UI(Y;X\backslash Z)$ are not equivalent as the next example shows. 
\begin{example}\label{ex:counterexample}
	Let $\Ycal=\{0,1,\mathtt{e}\}$, and $\Xcal=\Zcal=\{0,1,\mathtt{e}\}$. Let $P=P_Y\times P_{X|Y} \times P_{Z|X}$ where $P_Y\sim\text{Bernoulli}(\tfrac{1}{2})$ and $P_{X|Y}$ and $P_{Z|X}$ are symmetric erasure channels with erasure probabilities $\tfrac{1}{6}$ and $\tfrac{1}{5}$, resp. 
	Recall that a symmetric erasure channel from $\Ycal$ to $\Xcal$ with erasure probability $\epsilon\in[0,1]$ has transition probabilities: $P_{X|Y}(\mathtt{e}|0)=P_{X|Y}(\mathtt{e}|1)=\epsilon$, $P_{X|Y}(0|0)=P_{X|Y}(1|1)=1-\epsilon$.	For the distribution $P$, we have $UI(Y;X\backslash Z)=I(Y;X|Z)=\tfrac{1}{6}>0$.  
	On the other hand, the induced ``reverse'' erasure channels $P_{Y|X}=\kappa$ and $P_{Y|Z}=d$ are identical. Thus, $\delta^{\pi}(d,\kappa)=0$.
\end{example}
In \cite{e16042161},  
the value $UI(Y; X \backslash Z)$ is interpreted as the information about~$Y$ that is known to $X$ but unknown to~$Z$. 
This interpretation is motivated by Blackwell's result \cite{Blackwell1953}: whenever $UI(Y; X \backslash Z)>0$, there exists a decision problem in which it is better to know $X$ than to know $Z$. 
Moreover, this induces a decomposition of the mutual information between $Y$ and $(X,Z)$ into four terms:
\begin{align}\label{eq:MIdec1}
I(Y; XZ) &=UI(Y;X \backslash Z)+SI(Y; X,Z)\notag\\&+UI(Y;Z \backslash X)+CI(Y; X,Z).
\end{align}
%
The quantity 
\begin{equation}
\label{eq:MIdec2}
SI(Y;X,Z) := I(Y;X) - UI(Y;X\backslash Z)
\end{equation}
is interpreted as \emph{shared} or \emph{redundant} information, i.e., information about~$Y$ that is known in common to both $X$ and $Z$, and the quantity
\begin{equation}
\label{eq:MIdec3}
CI(Y;X,Z) := I(Y;X|Z) - UI(Y;X\backslash Z)
\end{equation}
is interpreted as \emph{complementary} or \emph{synergistic} information, i.e., the information about $Y$ that materializes only when $X$ and $Z$ act jointly.

\begin{example}
	If $X$ and~$Z$ are independent binary random variables, and $Y=\textsc{Xor}(X,Z)$, then $CI(Y;X,Z)=1$, while $SI(Y;X,Z) = UI(Y;X\backslash Z) = UI(Y;Z\backslash X)=0$. This is an instance of a purely \emph{synergistic} interaction. 
	
	If $Y$, $X$, $Z$ are uniformly distributed binary random variables with $Y = X = Z$, then $SI(Y;X,Z)=1$, while $CI(Y;X,Z) = UI(Y;X\backslash Z) = UI(Y;Z\backslash X)=0$.
	This is an instance of a purely \emph{redundant} interaction.  
	
	If $X$, $Z$ are independent binary random variables, and $Y = (X, Z)$, then $UI(Y;X\backslash Z) = UI(Y;Z\backslash X)=1$, while $SI(Y;X,Z)=CI(Y;X,Z)=0$. This is an instance of a purely \emph{unique} interaction. 	
	
\end{example}
For some probability distributions $P$ with special structure, the decomposition \eqref{eq:MIdec3} can be computed analytically \cite{e16042161}. 
\begin{lemma}\label{lem:UI_CMI}
	Let $Q^0:=P_{YX}\times e_{Z|X}$ for some $e\in\mathsf{M}(\Xcal;\Zcal)$. Then $UI_{Q^0}(Y;X\backslash Z) = I_{Q^0}(Y;X|Z)$, $SI_{Q^0}(Y;X,Z)=I(Y;Z)$ and $UI_{Q^0}(Y;Z\backslash X)=CI_{Q^0}(Y;X,Z)=0$.
\end{lemma}
\begin{proof}
	The distribution $Q^0$ defines a Markov chain $Y-X-Z$, which implies that $$I_{Q^0}(Y;Z|X)=0\le \min_{Q\in\Delta_P} I_Q(Y;Z|X).$$
	Hence, $Q^0$ solves the optimization problem~\eqref{subeq:UIy}.
\end{proof}
In the setting of the Lemma, $CI_{Q^0}(Y;X,Z)=0$ and therefore $Q^0$ is a \emph{zero-synergy distribution}. 

The information decomposition leads us to a new interpretation of the IB as a \emph{Unique Information Bottleneck}. 
To see this we first make the following definition. 
\begin{definition}[IB curve]\label{def:IBcurve}
	The IB curve is defined as follows \cite{witsenhausen1975,Harremoes2007,Tishbygilad}:
	\begin{align}
	B(r):= \max \{I(Z;Y):I(Z;X)\le r, \text{ }r\ge 0\}. 
	\end{align}
	Here the maximization is over all random variables $Z$ satisfying the Markov condition $Y-X-Z$ with fixed $(X,Y)\sim P$ and it suffices to restrict the size of $\Zcal$ to $|\Xcal|$. 
\end{definition}
The IB curve is concave and monotonically nondecreasing~\cite{witsenhausen1975,Tishbygilad}. 
We can explore the IB curve by solving the following optimization problem:
\begin{align}\label{eq:IBobjective}
\min_{e\in\mathsf{M}(\Xcal;\Zcal)} \Big[ UI_{Q^0}(Y;X\backslash Z) + \beta I_{Q^0}(Z;X) \Big]. 
\end{align}
Here $Q^0$ is a zero-synergy distribution and $\beta\in[0,1]$ is a Lagrange multiplier. 
Equation \eqref{eq:IBobjective} has the flavor of a rate-distortion problem  \cite{Harremoes2007} where the term  $UI_{Q^0}(Y;X\backslash Z)=I_{Q^0}(Y;X|Z)$ is interpreted as the average \emph{distortion} and $I_{Q^0}(Z;X)$ as the \emph{rate}. 
Classically, the IB is formulated as the maximization of $I_{Q^0}(Y;Z)-\beta I_{Q^0}(Z;X)$ \cite{InformationBottleneck}. 
By Lemma~\ref{lem:UI_CMI}, we have that $UI_{Q^0}(Y;X\backslash Z)=I_{Q^0}(Y;X|Z)=I_{Q^0}(Y;X)-I_{Q^0}(Y;Z)$. 
Since $I_{Q^0}(Y;X)=I_P(X;Y)$ is constant, it follows that the Unique Information Bottleneck defined by \eqref{eq:IBobjective} is equivalent to the classical IB. 
Each point on the IB curve satisfies the Markov condition $Y-X-Z$ which implies that the solution is always constrained to have \emph{zero synergy} about the output $Y$. 

Like the $UI$, the deficiency also induces an information decomposition as we show next.  

\vspace{.2cm}
\subsubsection{Deficiency induces an information decomposition}\label{app:defdecomp}
\label{ssec:DB}
We first propose a general construction that forms the basis of an information decomposition satisfying~\eqref{eq:MIdec1}-\eqref{eq:MIdec3}  (proved in the Appendix):
\begin{proposition} \label{prop:constructdecomp}
	Let $(Y,X,Z)\sim P$ and let $\delta^X$ be a nonnegative function 
	defined on the simplex~$\mathbb{P}_{\Ycal\times\Xcal\times\Zcal}$ 
	that satisfies the bound:
	\begin{align} \label{subeq:delbounds}
		0\le \delta^X(P) \le \min\{I(Y;X),I(Y;X|Z)\}.
	\end{align}
	Let $\Xcal'=\Zcal$, $\Zcal'=\Xcal$, and define a function $\tau:\mathbb{P}_{\Ycal\times\Xcal\times\Zcal}\to \mathbb{P}_{\Ycal\times\Xcal'\times\Zcal'}$ such that $\tau(P_{YXZ}(y,x,z))$ $=P_{YX'Z'}(y,z,x)$. Let 
	$\delta^Z(P):=\delta^X(\tau(P))$.
	Then the following functions define a nonnegative information decomposition satisfying~\eqref{eq:MIdec1}-\eqref{eq:MIdec3}:
	\begin{subequations}
		\label{subeq:decompgeneral}
		\begin{align*}
			\TUI(Y;X\backslash Z) &= \max\{\delta^X,\delta^Z+I(Y;X)-I(Y;Z)\},
			\\
			\TUI(Y;Z\backslash X) &= \max\{\delta^Z,\delta^X+I(Y;Z)-I(Y;X)\},
			\\
			\TSI(Y;X,Z) &=\min\{I(Y;X)-\delta^X, I(Y;Z)-\delta^Z\},
			\\
			\TCI(Y;X,Z) &= \min\{I(Y;X|Z)-\delta^X, I(Y;Z|X)-\delta^Z\}.
		\end{align*}
	\end{subequations}
\end{proposition}

We now apply the construction in Proposition~\ref{prop:constructdecomp} to derive an information decomposition based on the deficiency. 
The following proposition is proved in the Appendix. 
\begin{proposition}\label{lem:positivity_gdefi}
	Let~$(Y,X,Z)\sim P$, and 
	let $\kappa\in\mathsf{M}(\Xcal; \Ycal)$ and $d\in\mathsf{M}(\Zcal; \Ycal)$ be two channels 
	representing, resp., the conditional distributions $P_{Y|X}$ and $P_{Y|Z}$.
	Define $\delta^X=\delta^{\pi}(d,\kappa)$.
	Then the functions $\TUI$, $\TSI$, and $\TCI$ in Proposition~\ref{prop:constructdecomp} define a nonnegative information decomposition. 
\end{proposition}

The next proposition shows the relationship between the decompositions induced by the deficiency (see Proposition~\ref{lem:positivity_gdefi}) and that induced by the $UI$ (see \eqref{subeq:UIy}--\eqref{eq:MIdec3}). 
\begin{proposition}[\hspace{-.1mm}{{\cite{e16042161}}}]\label{lem:minsyn}
	\begin{align*}
		\TUI(Y;X\backslash Z) \le UI(Y;X\backslash Z),\\
		\TUI(Y;Z\backslash X) \le UI(Y;Z\backslash X),\\
		\TSI(Y;X,Z) \ge SI(Y;X,Z),\\
		\TCI(Y;X,Z) \ge CI(Y;X,Z),
	\end{align*}
	with equality if and only if there exists~$Q\in\Delta_P$ such that~$\TCI_{Q}(Y;X,Z)=0$.
\end{proposition}

\vspace{.3cm}
\section{Experiments}
\label{sec:experiments}

\subsection{Experiments on \textsc{MNIST}}
We present experiments on the \textsc{MNIST} dataset.
Classification on \textsc{MNIST} is a very well-studied problem. The main objective of these experiments is to evaluate the information-theoretic properties of the representations learned by the VDB model and to compare the classification accuracy for different values of $M$, the number of encoder output samples used in the training objective~\eqref{eq:emp-VDB-obj} when using the \emph{oneshot} strategy. As mentioned in Section~\ref{sec:method}, when $M=1$, we recover the VIB model~\cite{alemi2016variational}.

\textbf{Settings.} 
For the encoder, we use a fully connected feedforward network with 784 input units--1024 ReLUs--1024 ReLUs--512 linear output units. The deterministic output of this network is interpreted as the vector of means and variances of a 256-dimensional Gaussian distribution. The decoder is simply a softmax with $10$ classes. These are the same settings of the model used by~\cite{alemi2016variational}. 
At test time, the classifier is evaluated using $L$ encoder samples (i.e., we use $\frac{1}{L}\sum_{j=1}^L d(y|z^{(j)})$ where $z^{(j)}\sim e(z|x)$). 
We implement the algorithm in TensorFlow and train for $200$ epochs using the Adam optimizer. 

\textbf{Test accuracy.} 
The resulting test accuracy for different values of~$\beta$ and~$M$ is reported in Fig.~\ref{fig:exp1a}(a) and Table~\ref{tab:testacc}. 
As can be seen from 
Fig.~\ref{fig:exp1a}(a), the test accuracy is stable with increasing $M$. 
From Table~\ref{tab:testacc} we see that choosing $M$ larger than one can in fact slightly improve test accuracy. 
\begin{figure*}[h!]
	\centering
	\scalebox{.9}{	
		\begin{tabular}{cccc}
			\includegraphics[clip=true,trim=3.5cm 9.5cm 4cm 10cm,width=.26\textwidth]{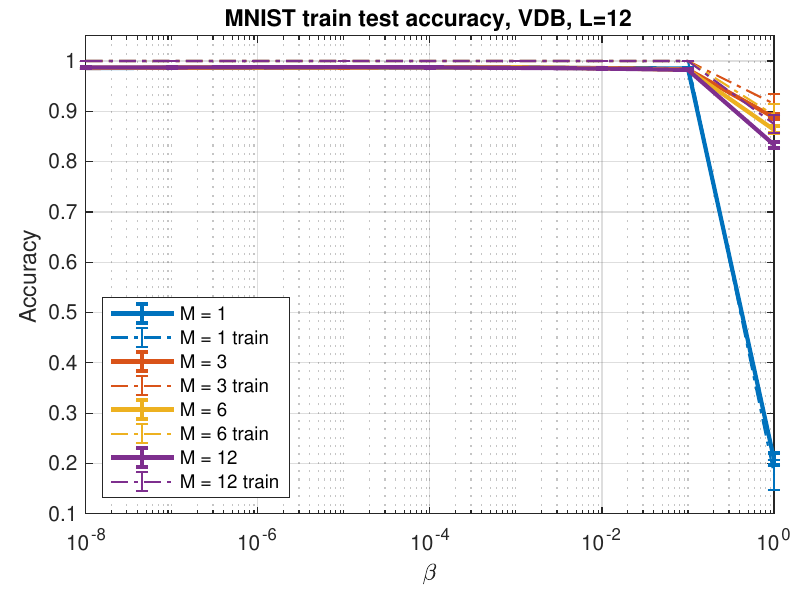}& 
			\includegraphics[clip=true,trim=3.5cm 9.5cm 4cm 10cm,width=.26\textwidth]{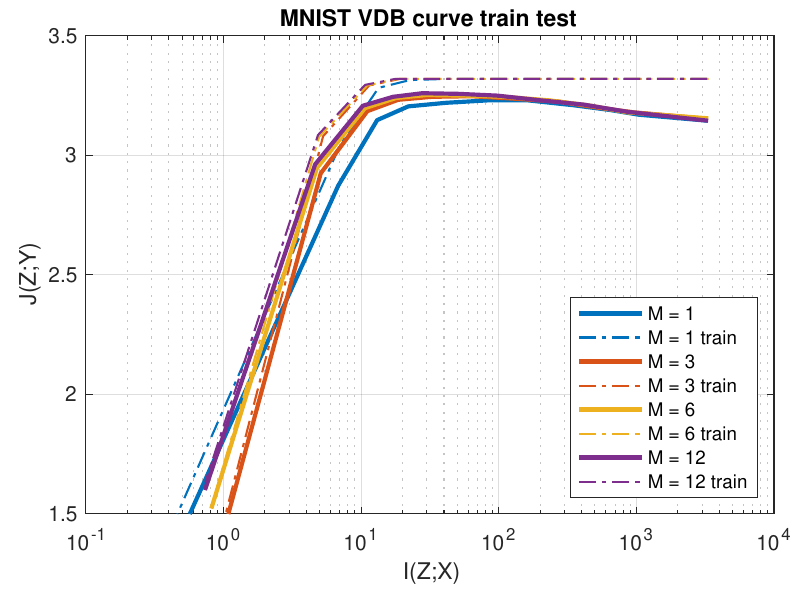}&
			\includegraphics[clip=true,trim=3.5cm 9.5cm 4cm 10cm,width=.26\textwidth]{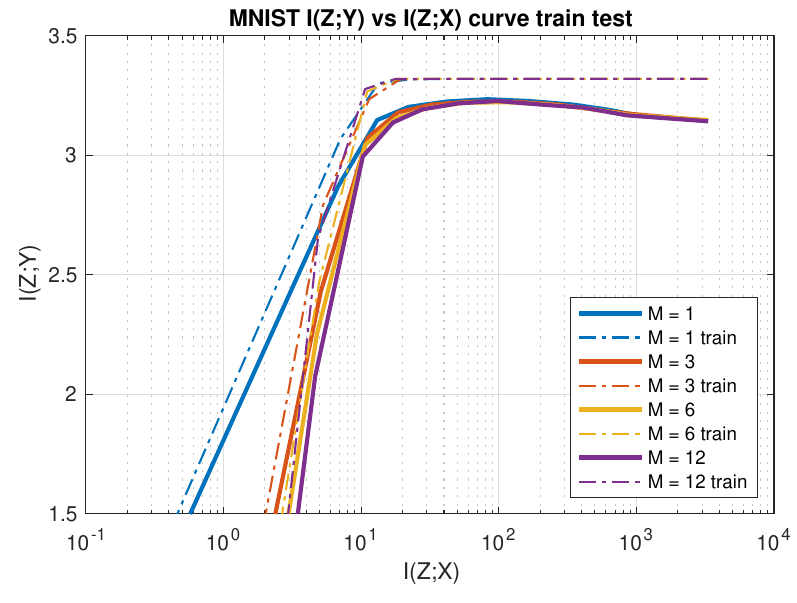}&
			\includegraphics[clip=true,trim=3.5cm 9.5cm 4cm 10cm,width=.26\textwidth]{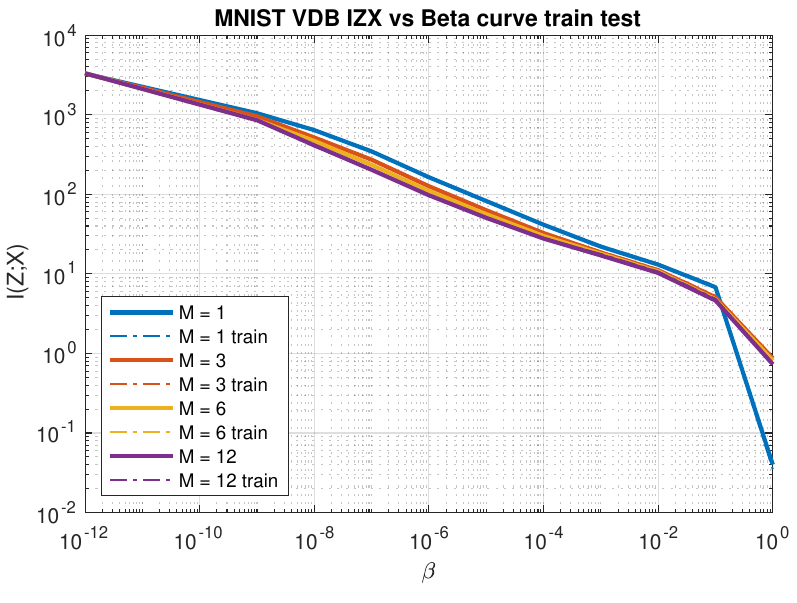}\\
			(a) & (b) & (c) & (d)
		\end{tabular}
	}	
	\caption{
		Effect of the regularization parameter $\beta$: (a) Accuracy on train and test data for the \textsc{MNIST} after training the VDB for different values of~$M$. Here $M$ is the number of encoder samples used in the training objective, and $L=12$ is the number of encoder samples used for evaluating the classifier. 
		(b) The VDB curve for different values of~$\beta$. The curves are averages over 5 repetitions of the experiment. Each curve corresponds to one value of $M=1,3,6,12$. 
		(c) The attained mutual information values $I(Z;Y)$ vs.\ $I(Z;X)$ for different values of~$\beta$.
		(d) $I(X;Z)$ vs.~$\beta$. 
		For~$M=1$, the VDB and the VIB models coincide. \\
	}
	\label{fig:exp1a}
\end{figure*}
\begin{table}[!b]
	\centering 
	\scalebox{1}{ 
		\begin{threeparttable}
			\captionsetup{justification=raggedright} 
			\caption{Test accuracy on \textsc{MNIST} for different values of~$\beta$ and~$M$, bottleneck size $K$, and $L= 12$. 
			} 
			\label{tab:testacc}
			\small 
			\begin{tabular}{cclcccc} 
				\toprule
				\multicolumn{1}{c}{$\beta$} & \multicolumn{1}{c}{$K$} & \multicolumn{4}{c}{$M$}\\
				\cmidrule{3-6}
				& & 1 & 3 & 6 & 12\\
				\midrule
				\num{e-05} & 256 & 0.9869 &  0.9873 & \textbf{0.9885} &  0.9878\\
				& 2              & 0.9575 &  0.9678 & \textbf{0.9696} &  0.9687\\
				\midrule
				\num{e-03} & 256 & 0.9872 &  0.9879 & 0.9875 &  \textbf{0.9882}\\
				& 2              & 0.9632 &  0.9726 & \textbf{0.9790} &  0.9702\\
				\bottomrule
			\end{tabular}
		\end{threeparttable}
	}		
\end{table} 

\textbf{Information curve.} 
The IB curve traces the mutual information~$I(Z;Y)$ of representation and output (sufficiency) vs.~the mutual information~$I(Z;X)$ of representation and input (minimality), for different values of the regularization parameter~$\beta$. 
In our method, ``more sufficient''  
is replaced by ``less deficient''. 
The term corresponding to sufficiency is $$J(Z;Y):=H(Y)-\mathbb{E}_{(x,y)\sim \ptrain} \left[-\log(\int d(y|z)e(z|x) \mathop{dz})\right].$$
Here $H(Y)$ is the entropy of the output, which for \textsc{MNIST} is $\log_2(10)$. 
Fig.~\ref{fig:exp1a}(b) shows the VDB curve which traces~$J(Z;Y)$ vs.~$I(Z;X)$ for different values of~$\beta$ at the end of training. 
Note that the curve corresponding to $M=1$ is just the VIB curve which traces~$I(Z;Y)$ vs.~$I(Z;X)$ for different values of~$\beta$.   
For orientation, lower values of $\beta$ have higher values of~$I(Z;X)$ (towards the right of the plot). 
For small values of~$\beta$, when the effect of the regularization is negligible, the bottleneck allows more information from the input through the representation. 
In this case, $J(Z;Y)$ increases on the training set, but not necessarily on the test set. 
This is manifest in the gap between the train and test curves indicative of a degradation in generalization. 
For intermediate values of~$\beta$, \emph{the gap is smaller for larger values of $M$ (our method)}. 
Fig.~\ref{fig:exp1a}(c) plots the attained mutual information values $I(Z;Y)$ vs.~$I(Z;X)$ after training with the VDB objective for different values of~$\beta$, 
while Fig.~\ref{fig:exp1a}(c) plots the minimality term~$I(Z;X)$ vs.~$\beta$. 
Evidently, the levels of compression vary depending on $M$. 
For good values of $\beta$, \emph{higher values of $M$ (our method) lead to a more compressed representation while retaining the same level of sufficiency}.
For example, for $\beta=10^{-5}$, setting $M=12$ requires storing $\sim 50$ less bits of information about the input when compared to the setting $M=1$, while retaining the same mutual information about the output.

\begin{figure*}[h!]
	\centering
	\resizebox {18cm}{!} {
		\begin{tikzpicture}[x=0.75pt,y=0.75pt,yscale=-1,xscale=1]
		\draw (58,1879) node  {\includegraphics[width=55.25pt,height=76.5pt]{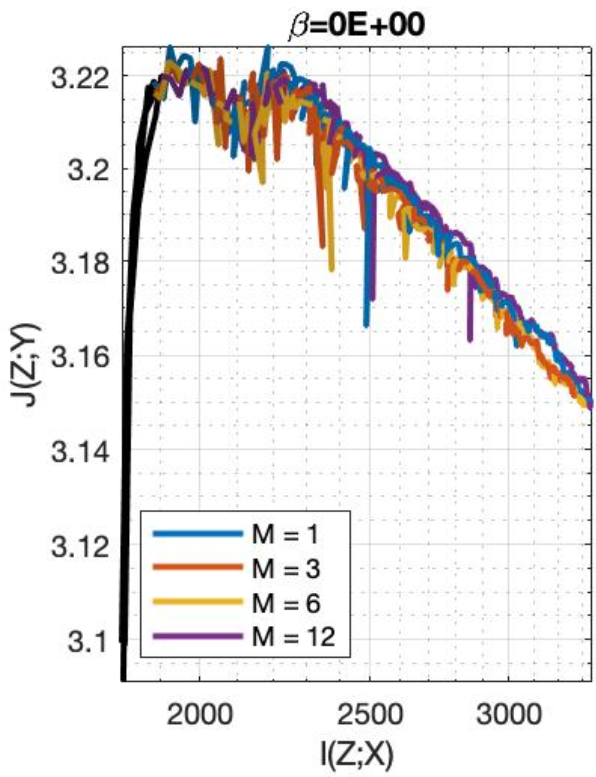}};
		\draw (136,1879) node  {\includegraphics[width=55.25pt,height=76.5pt]{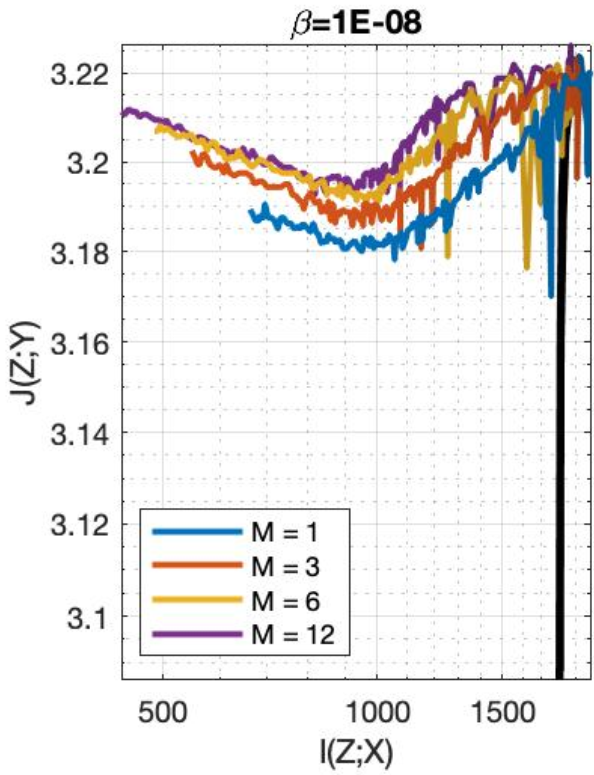}};
		\draw (292,1879) node  {\includegraphics[width=170.75pt,height=76.5pt]{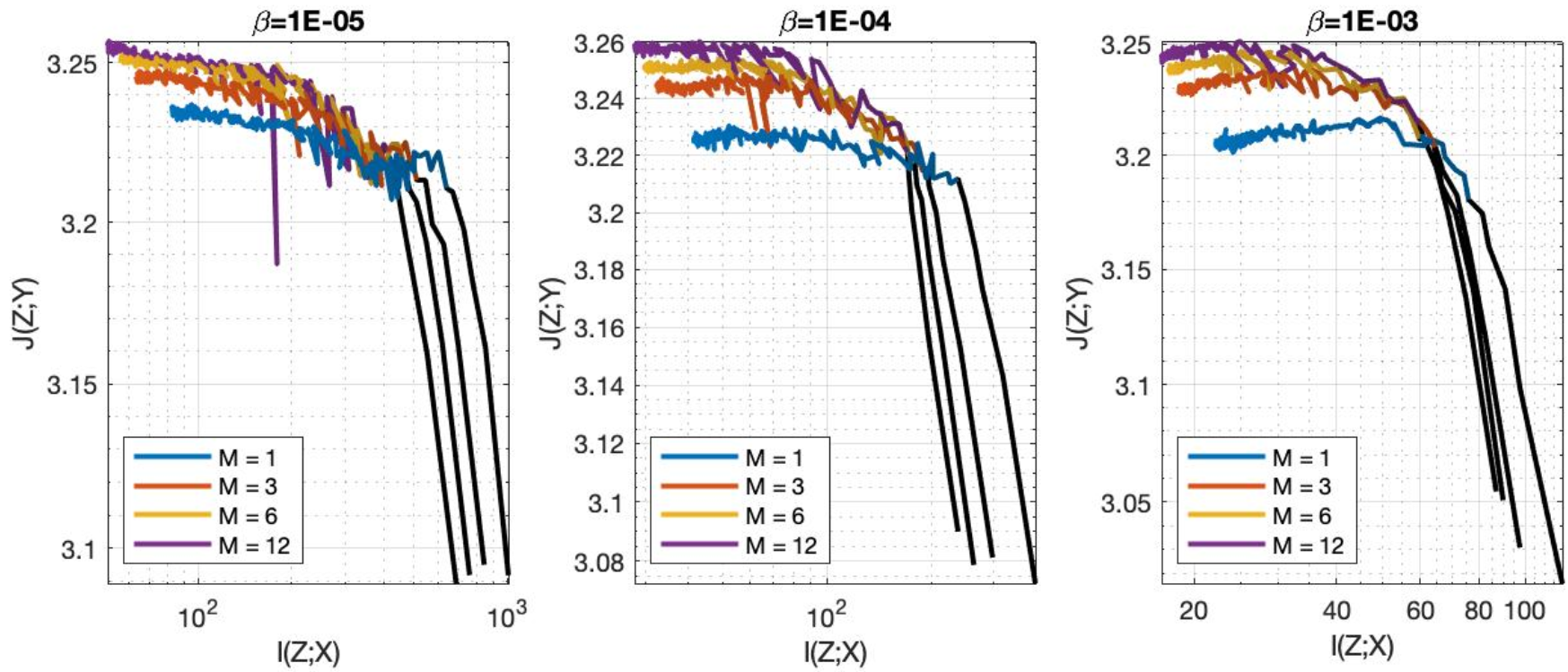}};
		\draw (488.25,1879) node  {\includegraphics[width=106.88pt,height=76.5pt]{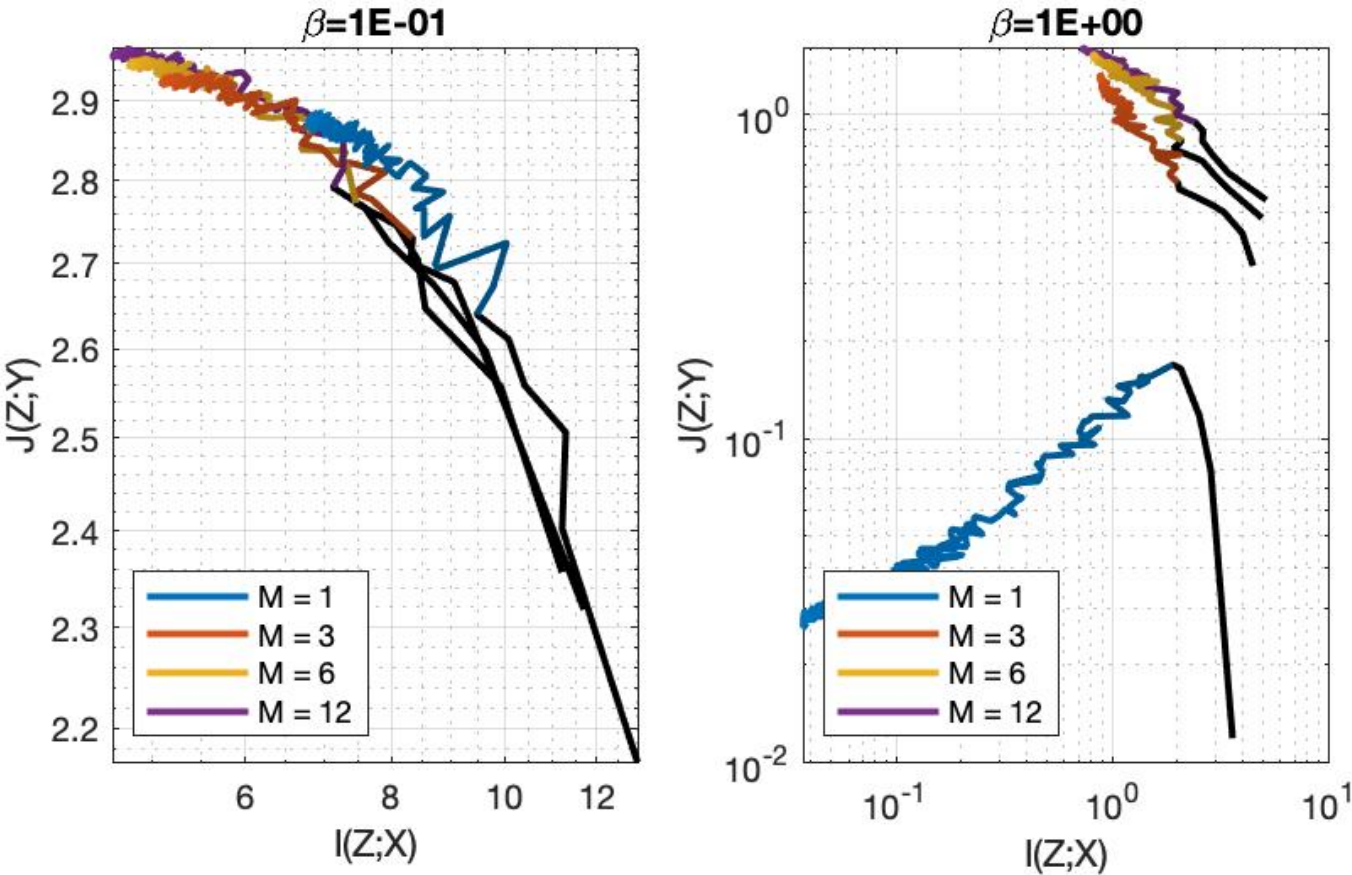}};
		\end{tikzpicture}}
	\caption{
		Evolution of the sufficiency and minimality terms (values farther up and to the left are better) over 200 training epochs (dark to light color) on \textsc{MNIST} with a 256-dimensional representation for different values of $\beta$. 
		The curves are averages over 20 repetitions of the experiment. $M=1$ corresponds to the VIB model. \\		
	} \label{fig:trainingdynamics}
\end{figure*}

\textbf{Training dynamics.} 
The dynamics of the information quantities during training are also interesting, and shown in Fig.~\ref{fig:trainingdynamics}. 
At early epochs, training mainly effects fitting of the input-output relationship and an \emph{increase of $J(Z;Y)$}. 
At later epochs, training mainly effects a \emph{decrease of $I(Z;X)$} leading to a better generalization. 
An exception is when the regularization parameter~$\beta$ is very small, in which case the representation captures more information about the input, and longer training decreases $J(Z;Y)$, which is indicative of overfitting. 
Higher values of~$M$ (our method) lead to the representation capturing more information about the output, while at the same time discarding more information about the input. $M=1$ corresponds to the VIB.

\textbf{Low dimensional representations.} 
To better understand the behavior of our method, we also visualize a 2-dimensional Gaussian representations after training the VDB with the \emph{oneshot} strategy for $M=1$ (when the VDB and VIB objectives are the same), and the \emph{sequential} strategy for $M=1$, and different values of the regularization parameter $\beta$. 
We use the same settings as before, with the only difference that the dimension of the output layer of the encoder is $4$, with two coordinates representing the mean, and two a diagonal covariance matrix. 
The results are shown in Fig.~\ref{fig:2dreps}. 
We see that for $\beta=10^{-4}$, representations of the different classes are well separated.
We also observe that as the frequency of the encoder updates is increased relative to the decoder (our method), individual clusters tend to be more spread out in latent space. 
This translates into a better discriminative performance when compared to the \emph{oneshot} (VIB) strategy as we show next with our robustness experiments. 
\begin{figure}[h!]
	\centering 
	\scalebox{1.1}{	
		\begin{small}	
			\begin{tabular}{cccc}
				$\beta$ & \emph{oneshot} (VIB) & $k=5$ (VDB) & $k=10$ (VDB) \\ \\
				\raisebox{.8cm}{$10^{-1}$} & 
				\includegraphics[width=1.75cm,clip=false,trim=0cm 0cm 0cm .7cm]{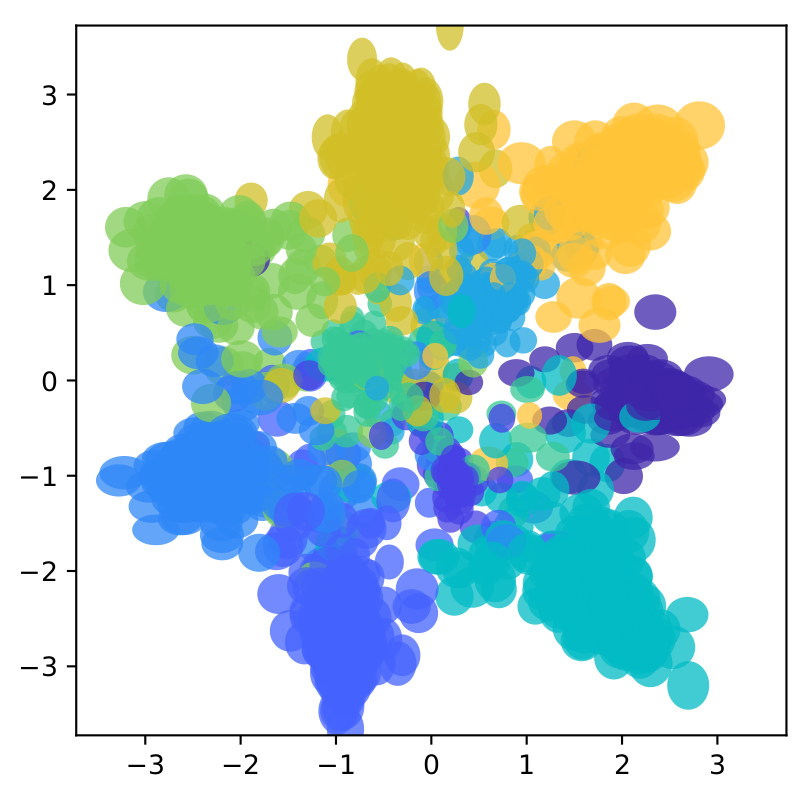} & 
				\includegraphics[width=1.75cm,clip=false,trim=0cm 0cm 0cm .7cm]{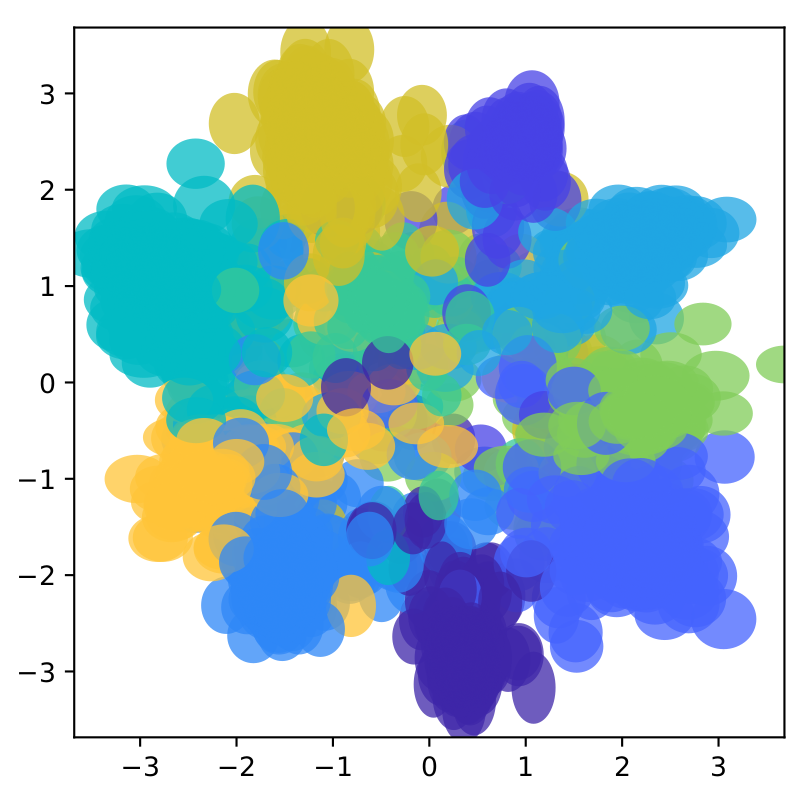} & 
				\includegraphics[width=1.75cm,clip=false,trim=0cm 0cm 0cm .7cm]{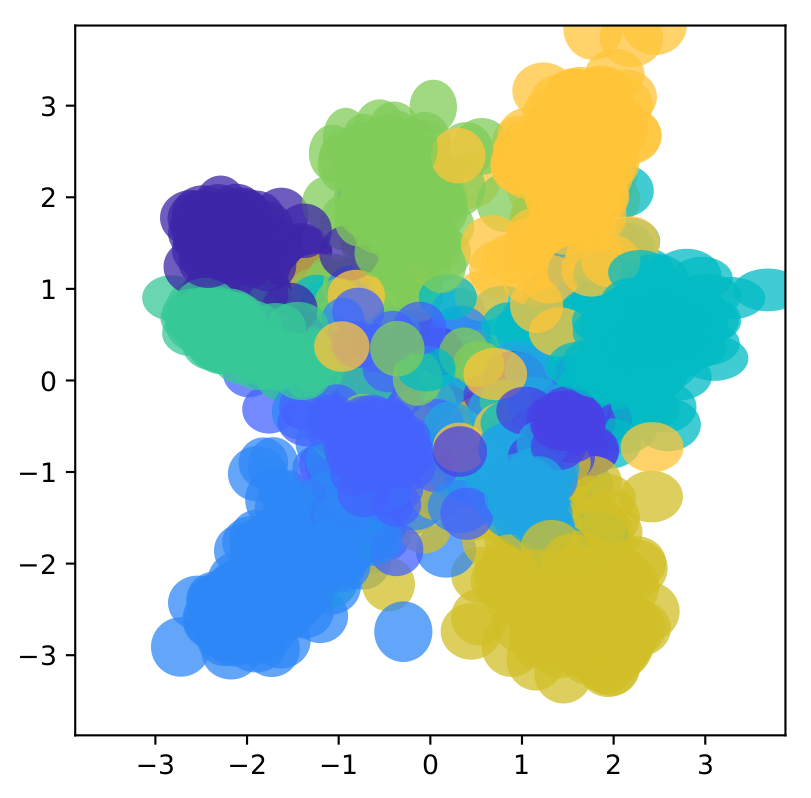} \\ 
				\raisebox{.8cm}{$10^{-3}$} & 
				\includegraphics[width=1.75cm,clip=false,trim=0cm 0cm 0cm .7cm]{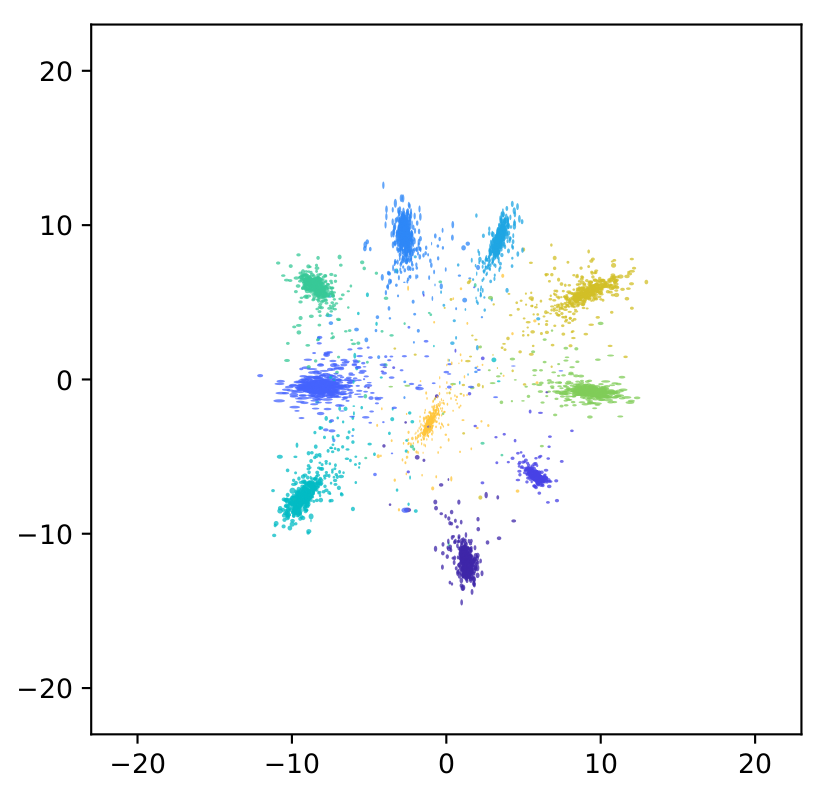} & 
				\includegraphics[width=1.75cm,clip=false,trim=0cm 0cm 0cm .7cm]{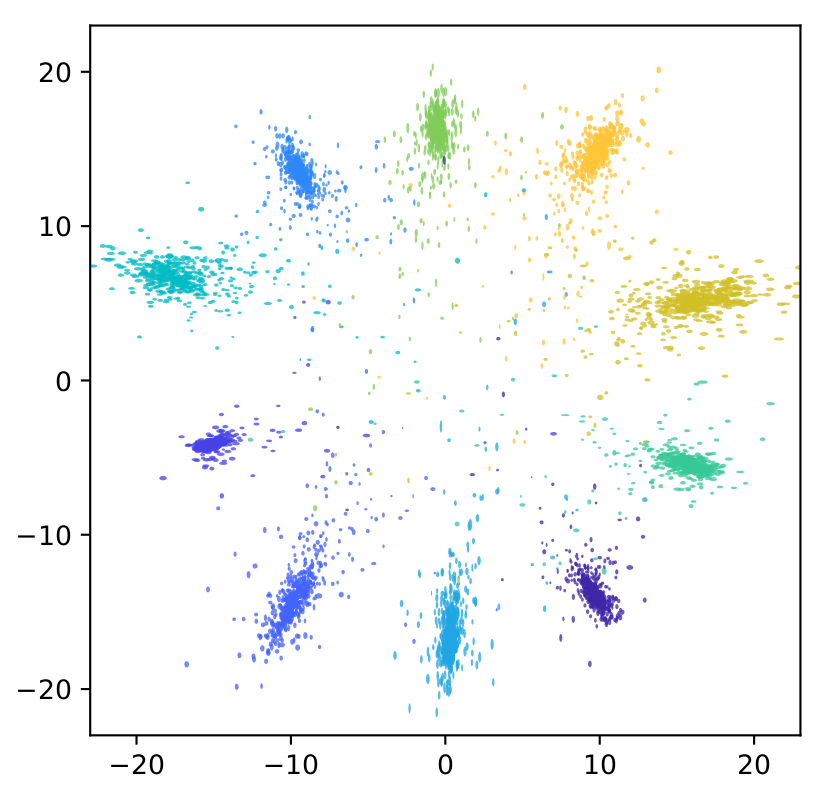} & 
				\includegraphics[width=1.75cm,clip=false,trim=0cm 0cm 0cm .7cm]{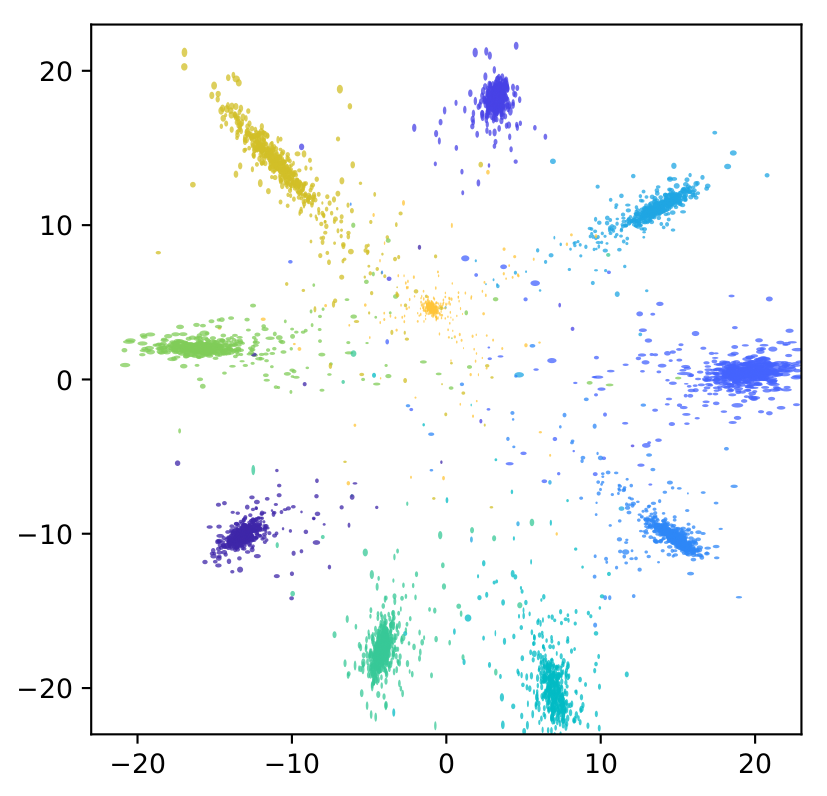} \\ 
				\raisebox{.8cm}{$10^{-4}$} & 
				\includegraphics[width=1.75cm,clip=false,trim=0cm 0cm 0cm .7cm]{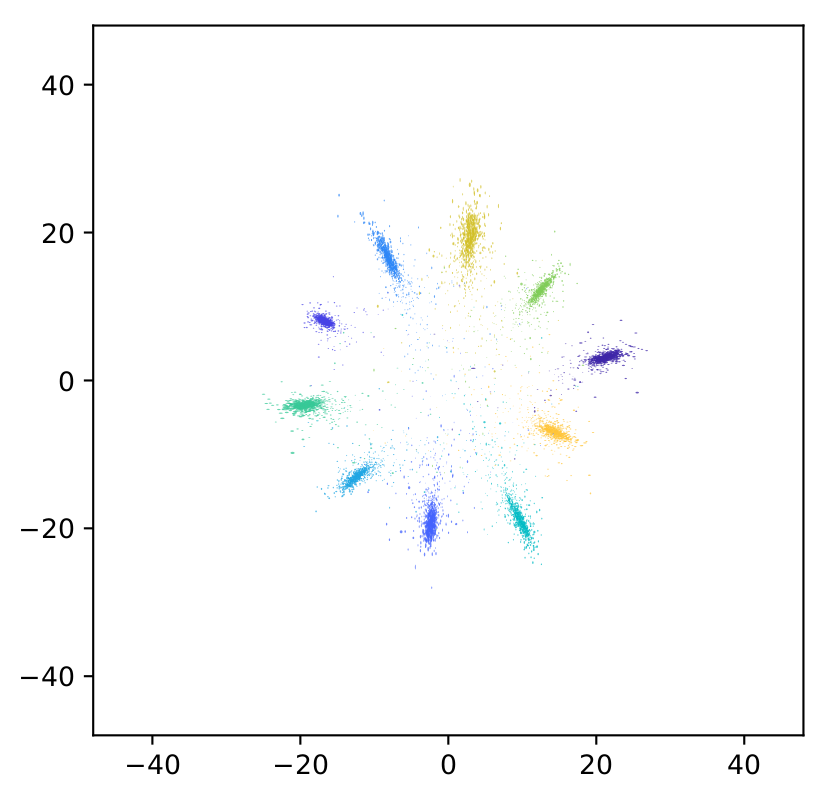} & 
				\includegraphics[width=1.75cm,clip=false,trim=0cm 0cm 0cm .7cm]{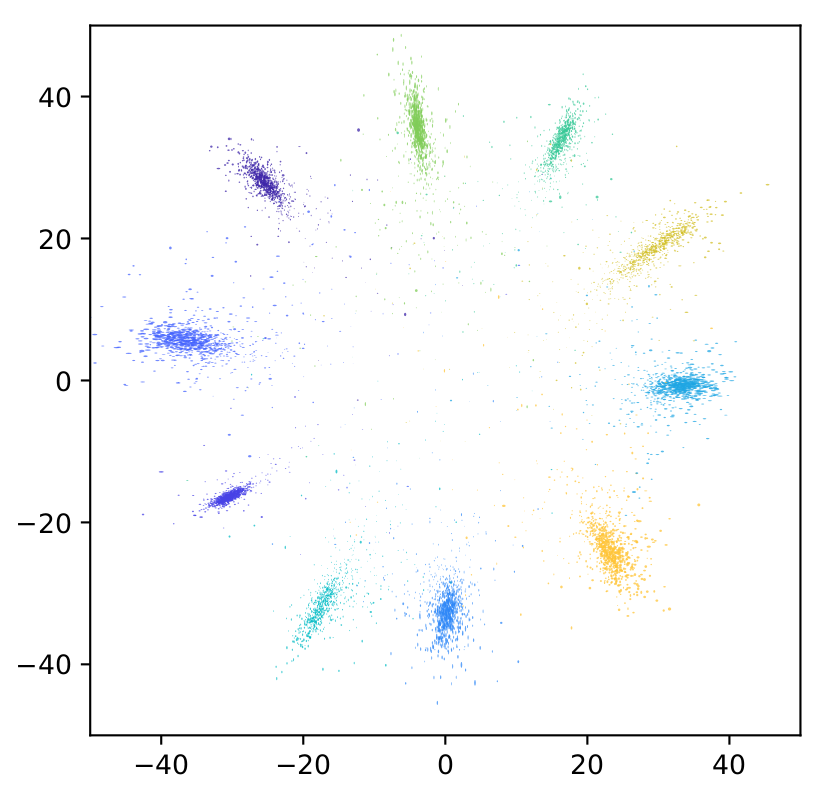} & 
				\includegraphics[width=1.75cm,clip=false,trim=0cm 0cm 0cm .7cm]{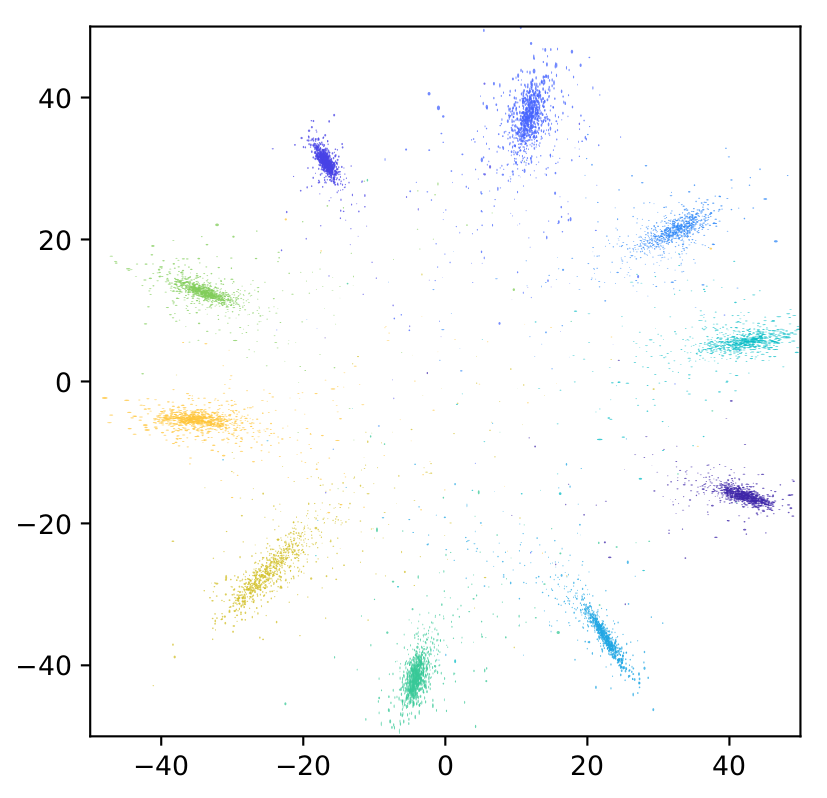} \\\\
			\end{tabular}
		\end{small}
	}	
	\caption{
		Posterior Gaussian distributions of 5000 test images from \textsc{MNIST} in a 2-dimensional latent space after training with $\beta=10^{-1},10^{-3},10^{-4}$ and $M=1$ with the \emph{oneshot} strategy and the \emph{sequential} strategy with $k=5,10$ encoder update steps per decoder update. Color corresponds to the class label. Boxes in each row have the same dimension.
		}
	\label{fig:2dreps}
\end{figure}

\subsection{Classification robustness under distributional shift}\label{sec:robustclass}
We demonstrate that the VDB generalizes well across distributional shifts, i.e., when the train and test distributions are different. 
We use the \textsc{MNIST-C} \cite{mnistc} and \textsc{CIFAR-10-C} \cite{augmixICLR2019} benchmarks to evaluate the classifier's robustness to common corruptions when trained with the VDB objective. 
These datasets are constructed by applying 15 common corruptions (at 5 different severity levels) to the \textsc{MNIST} and \textsc{CIFAR-10} test sets.
The corruptions comprise of four different categories, namely, noise, blur, weather, and digital.

To evaluate a classifier's robustness to common corruptions, we use the metrics proposed in \cite{augmixICLR2019}. Given a classifier $f$ and a baseline classifier $b$, the corruption error (CE) on a certain corruption type $c$ is computed as the ratio ${E_{c}^f}/{E_{c}^b}$, where $E_{c}^f$ and $E_{c}^b$ are resp. the errors of $f$ and $b$ on $c$, aggregated over five different severity levels.  
A more nuanced measure is the relative CE that measures corruption robustness relative to the \emph{Clean error}, the usual classification error on the uncorrupted test set. 
The relative CE is computed as the ratio ${(E_{c}^f-E_{clean}^f)}/{(E_{c}^b-E_{clean}^b)}$, where $E_{clean}^f$ and $E_{clean}^b$ are resp. the clean errors of $f$ and $b$. 
Averaging the CE and relative CE across all 15 corruption types yields the \emph{mean CE (mCE)} and the \emph{Relative mCE} values. 

Results of computation of the robustness metrics for the \textsc{MNIST-C} and \textsc{CIFAR-10-C} datasets using different training strategies for different values of $\beta$ are shown, resp., in Tables~\ref{tab:robusttestacc_mnist} and ~\ref{tab:robusttestacc_cifar}. The statistics are averages over $4$ independent runs.
The keys ``oneshot/M1'' and ``oneshot/M6'' refer to a \emph{oneshot} training strategy with resp., $M=1$ and $M=6$ encoder samples used for evaluating the training objective~\eqref{eq:emp-VDB-obj}. 
The keys ``seq:k:10/M1'' and ``seq:k:10/M6'' refer to a \emph{sequential} training strategy with resp., $M=1$ and $M=6$ encoder samples used for evaluating the training objective~\eqref{eq:emp-VDB-obj}, and with $k=10$ encoder update steps per decoder update.
We choose the baseline classifier as the VIB model (``oneshot/M1'').

For \textsc{MNIST-C}, we used the same encoder with a 256-dimensional representation as before.
We trained the VDB for 200 epochs using the Adam optimizer with a fixed learning rate of $10^{-4}$. 
For \textsc{CIFAR-10-C}, we used the 20-layer residual network ``ResNet20'' from \cite{resnet20} for the encoder with a 20-dimensional Gaussian representation and a softmax layer for the decoder. 
\begin{table}[!t]
	\centering 
	\scalebox{1}{ 
		\begin{threeparttable}
			\captionsetup{justification=raggedright} 
			\caption{\emph{Clean Error}, \emph{mCE}, and \emph{Relative mCE} values for the \textsc{MNIST-C} dataset using a MLP of size 784-1024-1024-512 trained using various strategies for different values of $\beta$. Lower values are better.
			} 
			\label{tab:robusttestacc_mnist}
			\small 
			\begin{tabular}{clcccc} 
				\toprule
				\multicolumn{1}{c}{$\beta$} & \multicolumn{1}{c}{Train strategy} & \multicolumn{1}{c}{\emph{Clean Error}} & \multicolumn{1}{c}{\emph{mCE}} & \multicolumn{1}{c}{\emph{Relative mCE}}\\
				\midrule
				\num{e-03} & oneshot/M1   &   1.53$\pm$0.10      & 100.00 &   100.00 \\
				           & seq/k:10/M1  &   1.42$\pm$0.13      &  93.88 &    93.56 \\
				           & oneshot/M6   &   1.38$\pm$0.04     &   86.84 &    84.46 \\
				           & seq/k:10/M6  &   1.39$\pm$0.04      &  85.95 &    \textbf{83.39} \\
				\midrule
				\num{e-04} & oneshot/M1   & 1.46$\pm$0.11       & 100.00 &   100.00 \\
				           & seq/k:10/M1  & 1.39$\pm$0.08       &  89.64 &    87.80 \\
						   & oneshot/M6   & 1.30$\pm$0.08        & 91.13 &   90.95 \\
						   & seq/k:10/M6  & 1.32$\pm$0.07        &  88.25 &    \textbf{86.85} \\
				\midrule
				\num{e-05} & oneshot/M1   &   1.63$\pm$0.44      & 100.00 &   100.00 \\
				           & seq/k:10/M1  &   1.32$\pm$0.11      &  84.64 &    \textbf{85.98} \\
				           & oneshot/M6   &   1.29$\pm$0.02      &  84.51 &    87.84 \\
				           & seq/k:10/M6  &   1.30$\pm$0.10      &  84.85 &    88.28 \\
				\bottomrule
			\end{tabular}
		\end{threeparttable}
	}		
\end{table} 

\begin{table}[t!]
	\centering 
	\scalebox{1}{ 
		\begin{threeparttable}
			\captionsetup{justification=raggedright} 
			\caption{\emph{Clean Error}, \emph{mCE}, and \emph{Relative mCE} values for the \textsc{CIFAR-10-C} dataset using the ResNet20 
			network \cite{resnet20}
			trained using various strategies for different values of $\beta$. Lower values are better.
			} 
			\label{tab:robusttestacc_cifar}
			\small 
			\begin{tabular}{clcccc} 
				\toprule
				\multicolumn{1}{c}{$\beta$} & \multicolumn{1}{c}{Train strategy} & \multicolumn{1}{c}{\emph{Clean Error}} & \multicolumn{1}{c}{\emph{mCE}} & \multicolumn{1}{c}{\emph{Relative mCE}}\\
				\midrule
				\num{e-03} & oneshot/M1   &   19.23$\pm$1.74      & 100.00 &   100.00 \\
				           & seq/k:10/M1 & 19.75$\pm$0.51        &  98.85 &    98.43 \\
				           & oneshot/M6    & 20.08$\pm$1.04        &  {96.91} &    \textbf{95.95} \\
				           & seq/k:10/M6 & 18.20$\pm$1.48        &  97.46 &    97.82 \\
				\midrule
				\num{e-04} & oneshot/M1    & 20.69$\pm$1.79        & 100.00 &   100.00 \\
				           & seq/k:10/M1 & 20.41$\pm$1.39        &  99.50 &    99.55 \\
						   & oneshot/M6    & 19.55$\pm$0.44        & 100.51 &   101.14 \\
						   & seq/k:10/M6 & 20.88$\pm$0.39        &  98.48 &    \textbf{98.12} \\
				\midrule
				\num{e-05} & oneshot/M1    &   20.42$\pm$1.05      & 100.00 &   100.00 \\
						   & seq/k:10/M1 & 18.78$\pm$0.52        &  97.22 &    \textbf{97.76} \\
						   &	 oneshot/M6    & 19.10$\pm$1.02        &  97.58 &    98.07 \\
						   & seq/k:10/M6 & 19.48$\pm$1.12        &  97.68 &    98.01 \\
				\bottomrule
			\end{tabular}
		\end{threeparttable}
	}		
\end{table} 

We see that for $M=1$, the sequential training strategy achieves lower \emph{mCE} and \emph{Relative mCE} values than the VIB across different values of $\beta$ for both the \textsc{MNIST-C} and \textsc{CIFAR-10-C} datasets. 
Recall that the objective of using the sequential strategy is to better approximate the deficiency which involves an optimization over the encoder. The advantage of sampling the encoder multiple times ($M>1$) for each input sample during training is also evident for both the oneshot and sequential strategy. The improved robustness in this case might be explained by way of data augmentation in latent space.

\section{Discussion}
We have formulated a bottleneck method based on channel deficiencies. 
The deficiency of a decoder w.r.t. a given channel quantifies how well an input randomization of the decoder (by a stochastic encoder) can be used to approximate the given channel. 
The DB has a natural variational formulation which recovers the VIB in the limit of a single sample of the encoder output.  
Moreover, the resulting variational objective can be implemented as an easy modification of the VIB objective with little to no computational overhead.  
Experiments show that the VDB can provide advantages in terms of minimality while retaining the same discriminative capacity as the VIB. 
We demonstrated that training with the VDB improves out-of-distribution robustness over the VIB on two benchmark datasets, the \textsc{MNIST-C} and the \textsc{CIFAR-10-C}.

An unsupervised version of the VDB objective \eqref{eq:emp-VDB-obj} (for $\beta=1$) shares some superficial similarities with the Importance Weighted Autoencoder (IWAE) \cite{iwae} which also features a sum inside a logarithm. Note, however, that the IWAE objective cannot be decomposed for~$M > 1$.
This implies that we cannot trade-off reconstruction fidelity for learning meaningful representations by incorporating bottleneck constraints. 
As $M$ increases, while the posterior approximation gets better, the magnitude of the gradient w.r.t. the encoder parameters also decays to zero \cite{rainforthtighter}. 
This potentially limits the IWAE's ability to learn useful representations. 
It is plausible that a similar bias-variance trade-off occurs with the VDB objective for high values of $M$. This is worth investigating.

\appendix\label{app:proofs}

\begin{proof} [Proof of Proposition~\ref{prop:constructdecomp}]
	Nonnegativity of $\TUI$, $\TSI$ and $\TCI$ follows from \eqref{subeq:delbounds} and the fact that $$0\le \delta^Z \le \min\{I(Y;Z),I(Y;Z|X)\}$$ by assumption.
	
	If $I(Y;Z)-\delta^Z \le I(Y;X)-\delta^X$, or equivalently, by the chain rule of mutual information \cite{CoverThomas91:Elements_of_Information_Theory}, if $I(Y;Z|X)-\delta^Z\le I(Y;X|Z)-\delta^X$, then we have  
	\begin{align*}
	\TUI(Y;X\backslash Z) &= \delta^Z+I(Y;X)-I(Y;Z),\\
	\TUI(Y;Z\backslash X) &= \delta^Z,\\
	\TSI(Y;X,Z) &=I(Y;Z)-\delta^Z,\\
	\TCI(Y;X,Z) &=I(Y;Z|X)-\delta^Z.
	\end{align*}
	Clearly, the functions $\TUI$, $\TSI$ and $\TCI$ satisfy~\eqref{eq:MIdec1}-\eqref{eq:MIdec3}, and the proposition is proved.
	The proof for the case when $I(Y;Z)-\delta^Z \ge I(Y;X)-\delta^X$ is similar. 
\end{proof}

\begin{proof} [Proof of Proposition~\ref{lem:positivity_gdefi}]
	It suffices to show that the $\delta^{\pi}(d,\kappa)$ satisfies the bound~\eqref{subeq:delbounds}.
	
	Let~$e^{\ast}\in\mathsf{M}(\Xcal;\Zcal)$ achieve the minimum in~\eqref{def:deficiency}.  
	By definition, $P_{Y|X}=\kappa$, $P_{Y|Z}=d$ and $P_X=\pi$. 
	We have
	\begin{multline*}
	I(Y;X|Z)=\sum_x P(x)\sum_z P(z|x)D(P(y|x,z)||P(y|z))\\
	\ge \sum_x P(x)D\left(\sum_z P(z|x)P(y|x,z)||\sum_z P(z|x)P(y|z)\right)\\
	=D(P_{Y|X}\|P_{Y|Z}\circ P_{Z|X}|P_X)
	\ge D(\kappa\|d\,\circ\, e^{\ast}|\pi)
	=\delta^{\pi}(d,\kappa),
	\end{multline*}
	where the first inequality follows from the convexity of the KL divergence and the second inequality follows from the definition of~$e^{\ast}$.
	
	$\delta^{\pi}(d,{\kappa})\le I(Y;X)$ since
	\begin{align*}
	&I(Y;X)-\delta^{\pi}(d,{\kappa})\\&=D(P_{Y|X}\|P_Y|P_X)-D(P_{Y|X}\|P_{Y|Z} \circ e^{\ast}_{Z|X} |P_X)\\
	&\ge D(P_{Y|X}\|P_Y|P_X)-D(P_{Y|X}\|P_{Y|Z} \circ P_{Z} |P_X)=0.
	\end{align*} 
\end{proof}

\section*{Acknowledgment}

PB thanks Pattarawat Chormai for many helpful inputs.

\bibliographystyle{IEEEtran}
\bibliography{IEEEabrv,general}

\end{document}